\documentclass{article}

\usepackage[margin=1in,footskip=0.25in]{geometry}
\usepackage[dvipsnames,svgnames,table]{xcolor}

\usepackage{amsfonts,amssymb,amsthm,amsmath}
\usepackage{bbm}
\usepackage{mathtools}
\usepackage{nicefrac}      \usepackage{commath}
\usepackage{nicematrix}
\usepackage{xspace}

\usepackage[sort&compress,square,comma,numbers]{natbib}
\usepackage{breakcites}

\usepackage{graphicx}
\usepackage{booktabs}
\usepackage{subcaption}

\usepackage{enumitem}
\usepackage[most]{tcolorbox}
\usepackage{algorithm}
\usepackage{algpseudocode}
\usepackage[labelfont=bf,labelsep=period]{caption}

\usepackage[normalem]{ulem} \usepackage{contour}
\usepackage{todonotes}

\usepackage{crossreftools}

\usepackage[
  pagebackref,
  breaklinks,
  colorlinks,
  linkcolor=PineGreen,
  citecolor=Blue,
  filecolor=Purple,
  urlcolor=Purple
]{hyperref}

\usepackage[capitalize, nameinlink, noabbrev]{cleveref}

\renewcommand*{\backref}[1]{}
\renewcommand*{\backrefalt}[4]{\ifcase #1 (No citations.)\or
    (Cited on page #2.)\else
    (Cited on pages #2.)\fi
}

\newcommand*{\email}[1]{\href{mailto:#1}{\nolinkurl{#1}} }

\algrenewcommand\algorithmicrequire{\textbf{Input:}}
\algrenewcommand\algorithmicensure{\textbf{Output:}}
\algrenewcommand\alglinenumber[1]{\sf\scriptsize\color{NavyBlue}{#1}}
\makeatletter
\renewcommand{\fnum@algorithm}{\textbf{\algorithmname~\thealgorithm.}}
\makeatother

\newcommand{\actionbox}[1]{\begin{tcolorbox}[colback=white,colframe=black,width=\columnwidth,boxsep=5pt,arc=4pt]
    \emph{#1}
\end{tcolorbox}}

\newcommand{\pad}{\vspace{1em}\hrule \vspace{1em}}

\contourlength{0.8pt}
\newcommand{\unline}[1]{\uline{\phantom{#1}}\llap{\contour{white}{#1}}}

\makeatletter
\def\hlinewd#1{\noalign{\ifnum0=`}\fi\hrule \@height #1 \futurelet
	\reserved@a\@xhline}
\makeatother

\newtheorem{theorem}{Theorem}

\newtheorem{lemma}{Lemma}
\newtheorem{fact}{Fact}

\newtheorem{definition}{Definition}

\newtheorem{problem}{Problem}

\newtheorem{importedtheorem}{Imported Theorem}

\Crefname{setting}{Setting}{Settings}
\Crefname{problem}{Problem}{Problems}
\Crefname{importedtheorem}{Imported Theorem}{Imported Theorems}

\makeatletter
\newtheorem*{rep@theorem}{\rep@title}
\newcommand{\newreptheorem}[2]{\newenvironment{rep#1}[1]{\def\rep@title{\Cref{##1} Restated}\begin{rep@theorem}}{\end{rep@theorem}}}
\makeatother
\makeatletter
\newtheorem*{rep@lemma}{\rep@title}
\newcommand{\newreplemma}[2]{\newenvironment{rep#1}[1]{\def\rep@title{\Cref{##1} Restated}\begin{rep@lemma}}{\end{rep@lemma}}}
\makeatother
\newreptheorem{theorem}{Theorem}
\newreplemma{lemma}{Lemma}

\newcommand{\defeq}[0]{\ensuremath{\;{\vcentcolon=}\;}\xspace}

\let\norm\relax
\newcommand{\norm}[1]{\enVert[0]{#1}}

\DeclareMathOperator*{\argmin}{arg\,min}

\DeclareMathOperator*{\E}{\mathbb{E}}

\newcommand{\etal}{\text{et al.}\xspace}

\newcommand{\eps}[0]{\ensuremath{\varepsilon}}
\let\epsilon\eps

 \newcommand{\mat}[1]{\mathbf{#1}} \renewcommand{\vec}[1]{\boldsymbol{\mathrm{#1}}}   

 \newcommand{\bmat}[1]{\begin{bmatrix} #1 \end{bmatrix}}

\newcommand{\mA}{\ensuremath{\mat{A}}\xspace}

\newcommand{\mI}{\ensuremath{\mathbf{I}}\xspace}

\newcommand{\mS}{\ensuremath{\mat{S}}\xspace}

\newcommand{\mV}{\ensuremath{\mat{V}}\xspace}

\newcommand{\mX}{\ensuremath{\mat{X}}\xspace}
\newcommand{\mY}{\ensuremath{\mat{Y}}\xspace}

\newcommand{\va}{\ensuremath{\vec{a}}\xspace}
\newcommand{\vb}{\ensuremath{\vec{b}}\xspace}

\newcommand{\vx}{\ensuremath{\vec{x}}\xspace}

\newcommand{\cC}{\ensuremath{{\mathcal C}}\xspace}

\newcommand{\cL}{\ensuremath{{\mathcal L}}\xspace}

\newcommand{\cN}{\ensuremath{{\mathcal N}}\xspace}
\newcommand{\cO}{\ensuremath{{\mathcal O}}\xspace}

\newcommand{\cS}{\ensuremath{{\mathcal S}}\xspace}

\newcommand{\bbC}{\ensuremath{{\mathbb C}}\xspace}

\newcommand{\bbR}{\ensuremath{{\mathbb R}}\xspace}

\newcommand{\bbZ}{\ensuremath{{\mathbb Z}}\xspace}

\newcommand{\rP}{\ensuremath{\mathrm{P}}\xspace}

\begin{document}
\author{Chris Camaño\thanks{Department of Computing and Mathematical Sciences, California Institute of Technology, Pasadena, CA 91125 USA (\email{ccamano@caltech.edu}, \email{ram900@caltech.edu},  \email{kshu@caltech.edu})}\and Raphael A. Meyer\footnotemark[1] \and Kevin Shu\footnotemark[1]}

\title{Debiasing Polynomial and Fourier Regression}
\maketitle

\begin{abstract}
We study the problem of approximating an unknown function $f:\mathbb{R}\to\mathbb{R}$ by a degree-$d$ polynomial using as few function evaluations as possible, where error is measured with respect to a probability distribution $\mu$.
Existing randomized algorithms achieve near-optimal sample complexities to recover a $ (1+\varepsilon) $-optimal polynomial but produce \textit{biased} estimates of the best polynomial approximation, which is undesirable.

We propose a simple debiasing method based on a connection between polynomial regression and random matrix theory.
Our method involves evaluating $f(\lambda_1),\ldots,f(\lambda_{d+1})$ where $\lambda_1,\ldots,\lambda_{d+1}$ are the eigenvalues of a suitably designed random complex matrix tailored to the distribution $\mu$.
Our estimator is unbiased, has near-optimal sample complexity, and experimentally outperforms iid leverage score sampling.

Additionally, our techniques enable us to debias existing methods for approximating a periodic function with a truncated Fourier series with near-optimal sample complexity.
\end{abstract}

\section{Introduction}
\label{sec:intro}
We study active polynomial regression.
Let $\mu$ be a given distribution on $\mathbb{R}$, and suppose we are given oracle access to an unknown function $f: \bbR\to\bbR$.
Our goal is, using as few evaluations of \(f\) as possible, to approximately recover the best degree-$d$ polynomial approximation of \(f\) under \(\mu\), defined as

\[
    p^* \defeq \argmin_{\deg(p) \leq d} \E_{t\sim\mu}\bigl[|p(t) - f(t)|^2\bigr].
\]
Typically, \(\mu\) is taken to be the uniform distribution on \([-1,1]\) \cite{meyer2023near}, or the Gaussian distribution \cite{erdelyi2020}.
In particular, we want to recover a polynomial \(\hat p\) of degree-\(d\) that has near-optimal expected error while using as few evaluations of \(f\) as possible:
\begin{problem}
    \label{prob:bayesian}
    Fix a degree \(d\), error tolerance \(\eps > 0\), and probability distribution \(\mu\) on \bbR.
    Let \(f:\bbR\to\bbR\) be a function we have oracle access to.
    Using as few evaluations \(f(t_1),\ldots,f(t_n)\) as possible, recover a polynomial \(\hat p\) such that
    \[
        \E_{t \sim \mu}\bigl[|\hat p(t) - f(t)|^2\bigr]
        \leq (1+\eps) \E_{t \sim \mu}\bigl[|p^*(t) - f(t)|^2\bigr].
    \]
\end{problem}

In \cref{sec:fourier} we also discuss a variant of \cref{prob:bayesian} tailored to identifying the low-degree Fourier coefficients of a periodic function with similar error guarantees.
Prior works have designed randomized algorithms solving \cref{prob:bayesian} with high probability using \(n=\cO(d \log(d) + \frac d\eps)\) or fewer evaluations \cite{chen2019active,shustin2022semi,meyer2023near,adcock2024optimal, kane2017robust, shimizu24}.
These methods share the undesirable trait of being \emph{biased}.
That is, they return a polynomial approximation \(\hat p\) with \(\E[\hat p] \neq p^*\).
Bias is undesirable because it limits which downstream algorithms we can use (e.g.\ model averaging), and because it can introduce unintuitive sources for error.
In this paper, we show that by evaluating \(f\) at the \unline{eigenvalues of a random Hermitian matrix}, we can remove this bias.

\begin{figure}[t]
    \centering
    \includegraphics[width=0.85\linewidth]{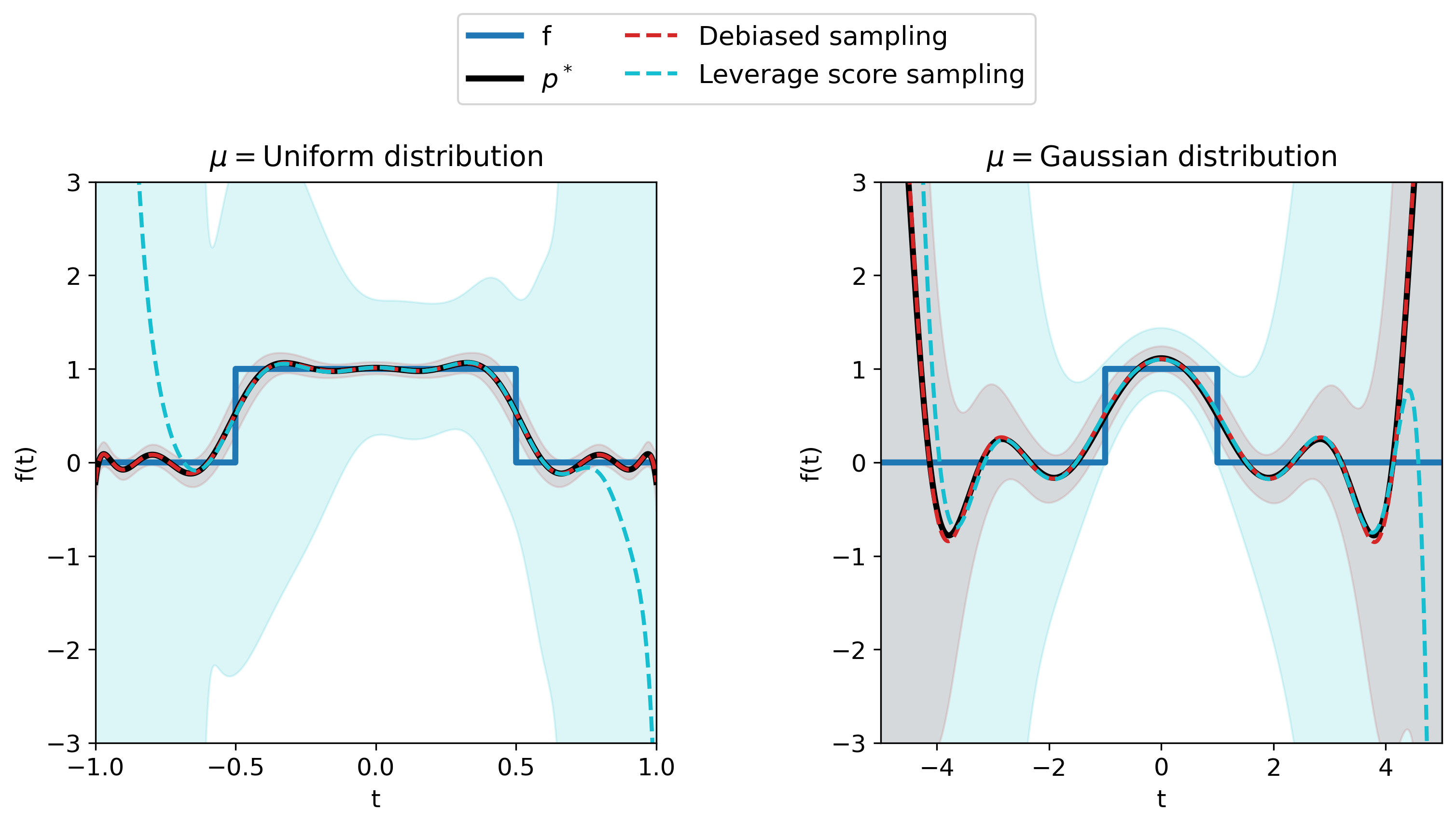}
    \caption{
    \textbf{Polynomial regression without bias.~} We run our debiased method \cref{alg:debiased-generic} (red) and leverage score sampling (blue) to find degree \(d=15\) polynomial approximations of indicator functions under different distributions \(\mu\), using \(n=35\) samples and repeated for 100,000 trials.
    The dotted lines show the empirical mean of returned polynomials across all independent trials, alongside the true best-fit polynomial \(p^*\) (black).
    Our method exactly recovers the optimal polynomial on average while leverage score sampling does not.
    Shaded areas show \(\pm1\) standard deviation bands.
    Our method also has significantly lower variance than leverage score sampling.
    \textit{Left:}
    When \(\mu\) is the uniform distribution on \([-1,1]\), we take \(f(t) = \mathbbm1_{t\in[-0.5,0.5]}\).
    \textit{Right:}
    When \(\mu\) is the Gaussian distribution \(\cN(0,1)\), we take \(f(t) = \mathbbm1_{t\in[-1,1]}\).
    }
    \label{fig:unbiased}
\end{figure}

We demonstrate the effect of our debiasing procedure in \cref{fig:unbiased}, where we compare our method against an existing method (leverage score sampling) which fails to match the best polynomial \(p^*\) on average for the allotted number of samples.
In contrast, our debiased sampling procedure exactly recovers \(p^*\) on average.
Later in \cref{sec:experiments}, we further observe that when \(n\) is small, leverage score sampling consistently produces an approximation of $p^*$ with \emph{orders of magnitude} larger error than our debiased method.
This is important in active learning scenarios like \cref{prob:bayesian}, where sample acquisition is expensive.
For instance, evaluating \(f(t)\) may amount to running a high-fidelity simulation or a real physical experiment.

To express our main result, it is helpful to understand leverage score sampling\footnote{
    Leverage score sampling is typically defined outside of polynomial regression \cite{shustin2022semi,avron2019universal,musco2022active,chen2019active}.
    Various papers connect the general definition and our definition \cite{shustin2022semi,pauwels2018relating,meyer2023near}; see \cite[Sec.~4.1]{meyer2023near} for a didactic proof.
}.
We define \(\rP_0,\ldots,\rP_d\) to be the orthonormal polynomials such that \(\langle\rP_i,\rP_j\rangle_\mu \defeq \E_{t\sim\mu}[\rP_i(t)\rP_j(t)]=\mathbbm1_{[i=j]}\).
In the context of \cref{prob:bayesian}, the \emph{leverage function} for \(\mu\) is \(\tau(t) \defeq \sum_{i=0}^d |\rP_i(t)|^2\) and the \emph{leverage score distribution}
is the distribution on \bbR with pdf \(q_{\rm lev}(t) = \frac1{d+1} \tau(t)\mu(t)\).
Here and throughout, we overload notation to let \(\mu\) denote both a distribution and its pdf.
Our main result is that a simple alteration of leverage score sampling solves \cref{prob:bayesian} in expectation for any continuous distribution \(\mu\), while being unbiased:

\begin{theorem}
    \label{thm:debiased-generic}
    Let \(\hat p\) be the result of running \cref{alg:debiased-generic}.
    Then if $n \ge d + 1$, then \(\hat p\) satisfies \(\E[\hat p] = p^*\).
    In addition, taking \(n = \cO(d \log(d) + \frac d\eps)\) samples suffices for $\hat p$ to solve \cref{prob:bayesian} in expectation:
    \[
        \E_{\hat p}\left[\E_{t\sim\mu}\bigl[|\hat p(t)-f(t) |^2\bigr]\right]
        \leq (1+\eps)\E_{t\sim\mu}\bigl[| p^*(t)-f(t)|^2\bigr].
    \]
\end{theorem}

\begin{algorithm}[t]
  \caption{Debiased active polynomial regression} \label{alg:debiased-generic}
  \begin{algorithmic}[1] 
    \Require Oracle access to $f$, target degree $d$, sample complexity \(n \geq d+1\), pdf \(\mu:\bbR\to\bbR\)
    \Ensure Degree-$d$ polynomial $\hat p$
\State Sample Hermitian \(\mX\in\bbC^{d+1 \times d+1}\) with pdf \(q_{\mu}(\mX) \propto \prod_{i=1}^{d+1} \mu(\lambda_i(\mX))\)
    \State Compute eigenvalues $t_1,\ldots,t_{d+1}\in\bbR$ of $\mX$
    \State Sample $t_{d+2},\ldots t_n$ iid with pdf $q_{\rm lev}(t)=\frac{1}{d+1}\mu(t)\tau(t)$
\State Evaluate \(f(t_i)\) for all \(i\in[n]\)
    \State Return polynomial \(\hat p = \argmin_{\deg(p) \leq d} \sum_{i=1}^n \frac1{\tau(t_i)} |p(t_i)-f(t_i)|^2\)
  \end{algorithmic}
\end{algorithm}

This theorem, proven in \cref{sec:thm_2_proof}, shows that we can augment samples from the leverage score distribution with the (real) \emph{eigenvalues of a random Hermitian matrix} and thereby produce an unbiased solution to the polynomial regression problem.
\Cref{thm:debiased-generic} follows from directly combining a result of Derezinksi, Warmuth, and Hsu \cite{derezinski2022unbiased} on solving general linear regression problems with a fact originating from Dyson \cite{dyson1970correlations} relating the eigenvalues of random matrices to determinantal point processes.

As stated, \cref{alg:debiased-generic} leaves several steps underspecified.
Line 5 requires solving a weighted least squares problem on finite data but over all polynomials of degree \(d\).
This can be solved efficiently with a variety of standard numerical methods \cite{trefethen2019approximation,brubeck2021vandermonde}.
Line 3 requires drawing samples from the leverage score distribution, which in principle can be done via rejection sampling as we can evaluate the pdf \(q_{\rm lev}(t)\) efficiently (see, e.g.\ \cite{devroye23}).

Most critically, on line 1, we need to generate a matrix \mX from the \emph{\(\mu\)-unitary distribution}.
That is, we need to sample a random Hermitian \(\mX\in\bbC^{d+1 \times d+1}\) with pdf \(q_\mu(\mX) \propto \prod_{i=1}^{d+1} \lambda_i(\mX)\).
A priori, it is not obvious how to do this.
One potential approach is to use MCMC methods, which can generate such an \mX for many distributions \(\mu\), though their mixing times are difficult to analyze.
However, simpler approaches are possible.
For a variety of distributions \(\mu\), including the Gaussian and uniform distributions, we entirely circumvent the issue of generating \mX and instead \unline{sample the eigenvalues of \mX \emph{without ever constructing \mX}}.

\section{Efficient algorithms for common measures}
\label{sec:tridiag}

In order to activate \cref{alg:debiased-generic}, we need to construct the eigenvalues of a random Hermitian matrix \(\mX\in\bbC^{d+1 \times d+1}\) drawn from the \(\mu\)-unitary ensemble.
For general distributions \(\mu\), it is not obvious how to generate the matrix \mX at all.
Nevertheless, a remarkable series of works on random matrix theory has shown that for many common distributions \(\mu\), we can easily sample a \unline{real symmetric tridiagonal} random matrix
\begin{equation}
    \mY = \bmat{
        \alpha_1 & \beta_1 \\
        \beta_1 & \alpha_2 & \beta_2 \\[0.25em]
        & \beta_2 & \smash\ddots & \smash\ddots \\
        & & \smash\ddots & \alpha_{d} & \beta_d \\
        & & & \beta_{d} & \alpha_{d+1}
    } \in \bbR^{d+1 \times d+1}
    \label{eqn:tridiag}
\end{equation}
with appropriately designed (random) entries \(\alpha_1,\ldots,\alpha_{d+1},\beta_1,\ldots,\beta_{d}\) such that the \emph{eigenvalues of \mY are distributed exactly as the eigenvalues of \mX}.
We say that \mY is sampled from the \emph{tridiagonal matrix model} for \mX.
These models have been described for a wide variety of distributions \(\mu\), and can be constructed in \(\cO(d)\) time.
Further, since we can compute the eigenvalues of a symmetric tridiagonal matrix like \mY in \(\cO(d \log d)\) time \cite{coakley13}, we recover a simple approach to run lines 1 and 2 of \cref{alg:debiased-generic}: construct this random tridiagonal \mY and return its eigenvalues.

To finish activating \cref{alg:debiased-generic}, we also need to draw samples from the leverage score distribution for \(\mu\).
While the aforementioned rejection sampling methods work (e.g.\ Devroye and Hamdan \cite{devroye23} sample from the leverage score distribution for Gaussian \(\mu\) in \(\cO(d^{2/3})\) time), we also provide a black-box approach:
\begin{fact}[Follows from \protect{\cite[Thm.~8]{derezinski2021determinantal}} and \cref{lem:proj-dpps-are-eigs}]
    \label{fact:dpp-to-levs}
    Let \(\mX\in\bbC^{d+1 \times d+1}\) be drawn from the \(\mu\)-unitary ensemble, and let \(i\) be drawn uniformly at random from \(\{1,\ldots,d+1\}\).
    Then \(\lambda_i(\mX)\) has pdf \(q_{\rm lev}(t)\).
\end{fact}

\begin{figure}
    \begin{algorithm}[H]
      \caption{Fast eigenvalue sampler}
      \label{alg:eigval-sampler}
      \begin{algorithmic}[1]
        \Require On-Diagonal values \(\alpha_1,\cdots,\alpha_k\in\bbR\), off-diagonal values \(\beta_1,\cdots,\beta_{k-1}\in\bbR\)
        \Ensure Eigenvalues $\lambda_1,\ldots,\lambda_k$
        \State Build symmetric tridiagonal \(\mY\in\bbR^{k\times k}\) with diagonal values \(\alpha_i\) and off-diagonal values \(\beta_i\).
        \State Return the eigenvalues $\{\lambda_i\}_{i=1}^k$ of $\mY$ \Comment{\textit{Using a symmetric tridiagonal eigensolver}}
      \end{algorithmic}
    \end{algorithm}
  \vspace{-1em}
    \begin{algorithm}[H]
      \caption{Fast leverage score distribution sampler}
      \label{alg:lev-sampler}
      \begin{algorithmic}[1]
        \Require On-Diagonal values \(\alpha_1,\cdots,\alpha_k\in\bbR\), off-diagonal values \(\beta_1,\cdots,\beta_{k-1}\in\bbR\)
        \Ensure Leverage score sample $\ell\in\bbR$
        \State Let \(\lambda_1,\ldots,\lambda_k\) be the result of running \cref{alg:eigval-sampler} with \(\alpha_1,\ldots,\alpha_k,\beta_1,\ldots,\beta_{k-1}\)
        \State Sample \(i\in\{1,\ldots,k\}\) uniformly at random
        \State Return \(\ell = \lambda_i\).
      \end{algorithmic}
    \end{algorithm}
  \vspace{-2em}
\end{figure}

So, we can generate each sample on line 3 of \cref{alg:debiased-generic} by sampling $\mY$ from an appropriate tridiagonal matrix model and returning a random eigenvalue of \mY, taking \(\cO(nd \log d)\) time to generate \(n\) samples from the leverage score distribution.
Note that other works on active polynomial regression \emph{approximately} sample from the leverage score distribution, which is faster to generate and achieves essentially the same asymptotic guarantees \cite{chen2019active,shustin2022semi,avron2019universal,meyer2023near}.
However, these approximations may inflate sample complexities by constant factors.
So, in order to optimize our empirical performance, we always sample from the \emph{exact leverage score distributions} in this paper.
We summarize our overall strategy with \cref{alg:eigval-sampler,alg:lev-sampler}.

\subsection{Random matrix selection}
We now outline a variety of distributions \(\mu\) and their corresponding tridiagonal matrix models \mY.
In this section, we always take \(\mX\in\bbC^{k \times k}\) to be drawn from the \(\mu\)-unitary ensemble.
We also take \(\mX\) to be a \(k \times k\) matrix to avoid any confusion or off-by-one errors; in \cref{alg:debiased-generic} we always take \(k=d+1\).
We then describe a tridiagonal matrix model \mY whose eigenvalues are distributed identically to those of \mX, with appropriate citations.
These methods allow \cref{alg:eigval-sampler,alg:lev-sampler} to run in \(\cO(k \log k)\) time.
We also mention which family of orthogonal polynomials is associated with each distribution, as these polynomials are needed to evaluate the leverage function on line 5 of \cref{alg:debiased-generic}.

We focus on the two most common distributions \(\mu\): the Gaussian distribution and the uniform distribution.
We note that other distributions, such as the Chi-squared \cite[Sec.~3.2]{dumitriu2002matrix} and Beta\footnote{
    This tridiagonalization result is originally proven by Killip and Nenciu \cite{killip2004matrix} who use a nonstandard definition of the Beta distribution.
    Hence, we refer the reader to \cite{duy2018spectral} for a more didactic reference.
} \cite[Sec.~5]{duy2018spectral} distributions have explicit tridiagonal matrix models which also lead to efficient sampling algorithms.
Some distributions, like the complex Gaussian distribution, do not have known tridiagonal matrix models but do admit direct construction of the matrix \mX \cite{byun2025progress}, allowing \cref{alg:eigval-sampler,alg:lev-sampler} to run in \(\cO(k^3)\) time.

\paragraph{Gaussian Distribution:~}
The setting where \(\mu=\cN(0,1)\) is the standard Gaussian distribution arises in a variety of applications \cite{tang2014discrete,hampton2015coherence,adcock2024learning,schwab2021deep,guo2020constructing} and yields the simplest tridiagonal construction.
In this case, the \(\mu\)-unitary ensemble has density \(q_{\mu}(\mX) \propto \prod_{i=1}^{d+1} e^{-|\lambda_i(\mX)|^2} = e^{-\norm\mX_{\rm F}^2}\), which is commonly referred to as the \emph{Gaussian Unitary Ensemble} (GUE) \cite{dumitriu2002matrix}.
For such matrices, Dumitriu and Edelman \cite{dumitriu2002matrix} show that building \mY in \cref{eqn:tridiag} with
\[
    \alpha_i \sim \cN(0,1)
    \quad
    \text{and}
    \quad
    \beta_i \sim \frac1{\sqrt 2} \chi_{2i},
\]
where \(\chi_k\) is the Chi distribution with parameter \(k\), ensures that the eigenvalues of \mY are distributed as those of \mX.
Further, the (normalized) \emph{probabilist's Hermite polynomials} form the orthonormal polynomials for \(\mu\).

\paragraph{Uniform Distribution:}
The setting where \(\mu\) is the uniform distribution on \([-1,1]\) is the most common situation that arises in practice, as it makes \cref{prob:bayesian} equivalent to solving the problem
\[
    \norm{\hat p-f}_{[-1,1]}^2
    \leq (1+\eps)
    \min_{\deg(p)\leq d} \norm{p-f}_{[-1,1]}^2,
\]
where \(\norm{f}_{[-1,1]}^2 = \int_{-1}^1 |f(t)|^2dt\), which has been explicitly studied in many prior works \cite{meyer2023near,kane2017robust,avron2019universal,shimizu24,adcock2018compressed,adcock2023fast}.
The definition of the \(\mu\)-unitary ensemble shows that \mX is distributed uniformly at random from the set of Hermitian matrices with eigenvalues in \([-1,1]\).
Dumitriu and Edelman \cite{dumitriu2002matrix} refer to the distribution of the matrix \(\frac12(\mX+\mI)\) as a Jacobi ensemble.
Killip and Nenciu \cite{killip2004matrix} (or \cite[Sec.~5]{duy2018spectral}; see footnote 2) constructed a tridiagonal matrix model for the Jacobi ensemble.
So, we can recover the tridiagonal matrix model for \mX by applying the affine transformation \(\mY \mapsto 2\mY-\mI\) to Killip and Nenciu's construction.
The matrix $\mY$ is defined in terms of independent random variables \(p_1,\ldots,p_{2k-1}\) with
\[
    p_i \sim
    \begin{cases}
        \operatorname{Beta}\bigl(
            \tfrac{2k-i}{2},\,\tfrac{2k-i+2}{2}
        \bigr) & i\text{ even} \\[6pt]
        \operatorname{Beta}\bigl(
            \tfrac{2k-i+1}{2},\,\tfrac{2k-i+1}{2}
        \bigr) & i\text{ odd}
    \end{cases}
\]
where \(p_{-1}=p_{0}=0\).
Using the format of \cref{eqn:tridiag}, the entries of \mY are then
\begin{align*}
    \alpha_i &=
    2\left(p_{2i-2}\bigl(1-p_{2i-3}\bigr)+p_{2i-1}\bigl(1-p_{2i-2}\bigr)\right)-1, \\
\beta_i &=
    2\left(\sqrt{\,p_{2i-1}\bigl(1-p_{2i-2}\bigr)p_{2i}\bigl(1-p_{2i-1}\bigr)\,}\right)-1,
\end{align*}
which ensures that the eigenvalues of \mY are distributed as those of \mX.
Further, the (normalized) \emph{Legendre polynomials} form the orthonormal polynomials for \(\mu\).

\section{Proof of \cref{thm:debiased-generic}: From regression to random matrices}\label{sec:thm_2_proof}

\Cref{thm:debiased-generic} says that we can solve \cref{prob:bayesian} by evaluating \(f\) at samples from the leverage score distribution for \(\mu\) and from the eigenvalues of a matrix drawn from the $\mu$-unitary ensemble.
This algorithm is a special case of \emph{leveraged volume sampling}, an algorithm proposed by Derezinksi, Warmuth, and Hsu \cite{derezinski2022unbiased} for general active linear regression problems.
So, to relate our problem to their setting, we first reformulate \cref{prob:bayesian} as a linear regression problem.

Let \(\rP_0,\ldots,\rP_d\) be the orthogonal polynomials with respect to \(\mu\), so that \(\rP_i\) has degree-\(i\) and \(\E_{t\sim\mu}[\rP_i(t)\rP_j(t)]=\mathbbm1_{[i=j]}\).
Then define the random vector
\[
    \va_t \defeq [\rP_0(t) ~ \cdots ~ \rP_d(t)]^\top\in\bbR^{d+1}
    \quad
    \text{where }
    t \sim \mu.
\]
Every degree-\(d\) polynomial \(p\) can be represented with a vector \(\vx\in\bbR^{d+1}\) by using the linear combination \(p(t) = \sum_{i=0}^d x_i \rP_i(t)\).
So, we can rewrite
\[
    \min_{\deg(p) \leq d} \E_{t \sim \mu} \bigl[|p(t) - f(t)|^2\big]
    =
    \min_{\vx\in\bbR^{d+1}} \E_{t \sim \mu} \bigl[|\va_t^\top\vx - f(t)|^2\bigr],
\]
which now has the appearance of a linear regression problem.

Derezinski \etal \cite{derezinski2022unbiased} define the \emph{leverage function}, \emph{leverage score distribution}, and \emph{projection DPP} for general linear regression problems.
Specializing these definitions to our polynomial regression setting, and simplifying with the observation that that \(\E[\va_t\va_t^\top]=\mI\), these definitions become:

\begin{definition}[Leverage function and distribution]
    \label{def:levs}
    The \unline{\emph{leverage function}} for measure \(\mu\) is \(\tau(t) \defeq \sum_{i=0}^d |\rP_i(t)|^2\).
    A variable \(t\) is drawn from the corresponding \unline{\emph{leverage score distribution}} if it has pdf \(q_{\rm lev}(t) \defeq \frac1{d+1} \mu(t)\tau(t)\).
\end{definition}

\begin{definition}[Projection DPP]
    \label{def:dpp}
    The \unline{\emph{projection DPP}} for measure a \(\mu\) is the distribution over sets \(\cS = \{t_1,\ldots,t_{d+1}\}\) with pdf \(q(\cS) = |\det(\mA_{\cS})|^2 \prod_{i=1}^{d+1}\mu(t_i)\) where \(\mA_{\cS} \defeq [\va_{t_1} ~ \cdots ~ \va_{t_{d+1}}]\in\bbR^{d+1 \times d+1}\).
\end{definition}

Note that Derezinki \etal's definitions recover the definitions of the leverage function and distribution we gave in \cref{sec:intro}.
Given these definitions, Derezinksi \etal \cite{derezinski2022unbiased} show the following:
\begin{importedtheorem}[Leveraged volume sampling solve linear regression; Thm.~3.1 from \cite{derezinski2022unbiased}]
    \label{impthm:leveraged-dpp}
    Fix degree \(d\), sample complexity \(n\geq d+1\), and error tolerance \(\eps > 0\).
    Sample \(\{t_1,\ldots,t_{d+1}\}\) from the projection DPP for \(\mu\).
    Sample \(t_{d+2},\ldots,t_n\) iid from the leverage score distribution for \(\mu\).
    Then the polynomial
    \[
        \hat p \defeq \argmin_{\deg(p) \leq d} \sum_{i=1}^n \frac{1}{\tau(t_i)}|p(t_i) - f(t_i)|^2
    \]
    has \(\E[\hat p] = p^*\).
    Further, taking \(n = \cO(d \log(d) + \frac d\eps)\) samples suffices to achieve
    \[
        \E_{\hat p}\left[\E_{t\sim\mu}\bigl[|\hat p(t) - f(t)|^2\bigr]\right]
        \leq (1+\eps)\E_{t\sim\mu}\bigl[|p^*(t) - f(t)|^2\bigr].
    \]
\end{importedtheorem}

The only mismatch between the statement of \cref{impthm:leveraged-dpp} and the statement of our main result \cref{thm:debiased-generic} lies in the construction of the DPP samples.
We use the well-known fact that the eigenvalues of certain random Hermitian matrices are distributed exactly as a projection DPP.
Formally, we invoke the following implication of Weyl's integration formula:
\begin{importedtheorem}[Weyl's integration formula, \protect{\cite[Eqn.~5.30]{livan2018introduction}}]
    \label{impthm:weyl-integral}
    Fix
pdf \(\mu:\bbR\to\bbR\).
    Let \(\mX\in\bbC^{k \times k}\) be drawn from the \(\mu\)-unitary ensemble.
    Then, the joint pdf of the eigenvalues of \mX is \(q(\lambda_1,\ldots,\lambda_k) \propto \prod_{i < j} |\lambda_i-\lambda_j|^2\prod_{i=1}^k \mu(\lambda_i) \).
\end{importedtheorem}

To prove \cref{thm:debiased-generic} using \cref{impthm:leveraged-dpp}, we now just show that the eigenvalue pdf in \cref{impthm:weyl-integral} matches the pdf of the projection DPP in \cref{def:dpp}:

\begin{lemma}[Polynomial regression projection DPPs are random matrix eigenvalues]
    \label{lem:proj-dpps-are-eigs}
    Let \(\mX\in\bbC^{d+1 \times d+1}\) be the random Hermitian matrix with pdf \(q_{\mu}(\mX) \propto \prod_{i=1}^{d+1} \mu(\lambda_i(\mX))\).
    Then, the eigenvalues of \mX are distributed as the projection DPP for \(\mu\).
    \label{lem:dpp_to_eigs}
\end{lemma}
\begin{proof}
    Let \(\alpha_i\) be the leading coefficient of polynomial \(\rP_i\), and let \(\alpha \defeq \prod_i \alpha_i\).
    Then, by the generalized Vandermonde determinant formula \cite[Eqn.~7.3]{livan2018introduction}, we have that
    \[
        \det(\mA_{\cS})
        = \alpha \prod_{i > j} (t_i - t_j).
    \]
    Therefore, the pdf of the projection DPP is 
    \begin{align*}
        q(\cS)
        = |\det(\mA_{\cS})|^2 \prod_{i=1}^{d+1}\mu(t_i)
        = \alpha^2 \prod_{i > j} |t_i - t_j|^2 \prod_{i=1}^{d+1}\mu(t_i).
    \end{align*}
    Invoking \Cref{impthm:weyl-integral} shows that the joint pdf of the eigenvalues of \mX is
    \[
        q(\{\lambda_1,\ldots,\lambda_{d+1}\})
        \propto \prod_{i < j} |\lambda_i-\lambda_j|^2 \prod_{i=1}^{d+1}\mu(\lambda_i),
    \]
    which exactly matches the pdf of \cS, completing the proof.
\end{proof}

\section{Low-Degree Fourier Coefficients of Periodic Functions}
\label{sec:fourier}

A variant of \cref{prob:bayesian} discusses approximating periodic functions using a Fourier series.
Recall that any periodic function \(g\) (say, with period \(2\pi\)) can be represented as a function defined on the unit circle in the complex plane \(\cC = \{z : |z|=1\}\) with the map \(f(e^{i\theta}) \defeq g(\theta)\).
For a function $f$ on $\cC$, the Fourier series expansion of $f$ is given by $f(e^{i\theta}) = \sum_{k=-\infty}^{\infty} c_k e^{ik\theta}$, or in other words $f(z) = \sum_{k=-\infty}^{\infty} c_k z^k$.
Further, \(\{z \mapsto z^k/\sqrt{2\pi} : k\in\bbZ\}\) is an orthonormal basis for \(\cL_2(\cC)\), the set of functions on \cC with finite norm \( \norm f_{\cC}^2 \defeq \int_{\cC} |f(z)|^2 dz\).
Therefore, recalling that \(c_k\) is the \(k\)-th Fourier coefficient of \(f\),
\[
    p^* = \argmin_{\deg(p) \leq d} \norm{p(z) - f(z)}_{\cC}^2 = 
    \sum_{k=0}^d c_k z^k.
\]
That is, the best degree-\(d\) polynomial approximation of \(f\) recovers the degree-\(d\) Fourier series expansion of \(f\).
This motivates the following problem statement:

\begin{problem}
    \label{prob:fourier}
    Fix a degree-\(d\), error tolerance \(\eps > 0\).
    Let \(f:\cC\to\bbC\) be a function we have oracle access to.
    Using as few evaluations \(f(z_1),\ldots,f(z_n)\) as possible, construct a polynomial \(\hat p\) such that
    \[
        \norm{\hat p - f}_{\cC}^2
        \leq (1+\eps) \norm{p^* - f}_{\cC}^2.
    \]
\end{problem}

By a similar analysis to that of \cref{thm:debiased-generic}, we show that \cref{alg:debiased-fourier} solves \cref{prob:fourier}:
\begin{theorem}
    \label{thm:debiased-fourier}
    Let \(\hat p\) be the result of running \cref{alg:debiased-fourier}. Then for any $n \ge d+1$, \(\hat p\) has \(\E[\hat p] = p^*\).
    In addition, taking \(n = \cO(d \log(d) + \frac d\eps)\) samples suffices for $\hat p$ to solve \cref{prob:fourier} in expectation:
    \[
        \E_{\hat p}\left[\norm{\hat p - f}_{\cC}^2\right]
        \leq (1+\eps)\norm{p^* - f}_{\cC}^2
    \]
\end{theorem}

Notice that \cref{alg:debiased-fourier} requires us to generate a uniformly (i.e. Haar) random unitary matrix \(\mX\in\bbC^{d+1 \times d+1}\) and compute its eigenvalues.
Naively, this takes \(\cO(d^3)\) time.
However, an algorithm of Fasi and Robol \cite{fasi2021sampling} shows that we can generate the eigenvalues of \mX directly in \(\cO(d^2)\) time.
An interesting open problem is to generate these eigenvalues in \(\widetilde \cO(d)\) time.
We also remark that other methods, such as querying \(f\) at points \(z_j = e^{\theta_0 + 2\pi j / n}\) for \(\theta_0 \in [0,2\pi]\) drawn uniformly at random, ensures that \(\E[\hat p] = p^*\).
Investigating the error \(\eps\) of this estimator in terms of \(n\) is an interesting open problem.

The proof of \cref{thm:debiased-fourier} benefits from the following classic result of Dyson, which serves as an analogue to \cref{impthm:weyl-integral} that considers eigenvalues on the complex circle \cC instead of the real line \bbR.
\begin{importedtheorem}
[\protect{\cite{dyson1970correlations}}, see also \protect{\cite[Thm.~3.1]{meckes2019random}}]
    \label{impthm:weyl-integral-fourier}
    Let \(\mu\) be the uniform distribution on \cC.
    Let \(\mX\in\bbC^{k \times k}\) be a uniformly (i.e. Haar) random unitary matrix.
    Then, the joint pdf of the eigenvalues of \mX is \(q(\lambda_1,\ldots,\lambda_k) \propto \prod_{i < j} |\lambda_i-\lambda_j|^2\).
\end{importedtheorem}

\begin{figure}
    \centering
    \includegraphics[width=.85\linewidth]{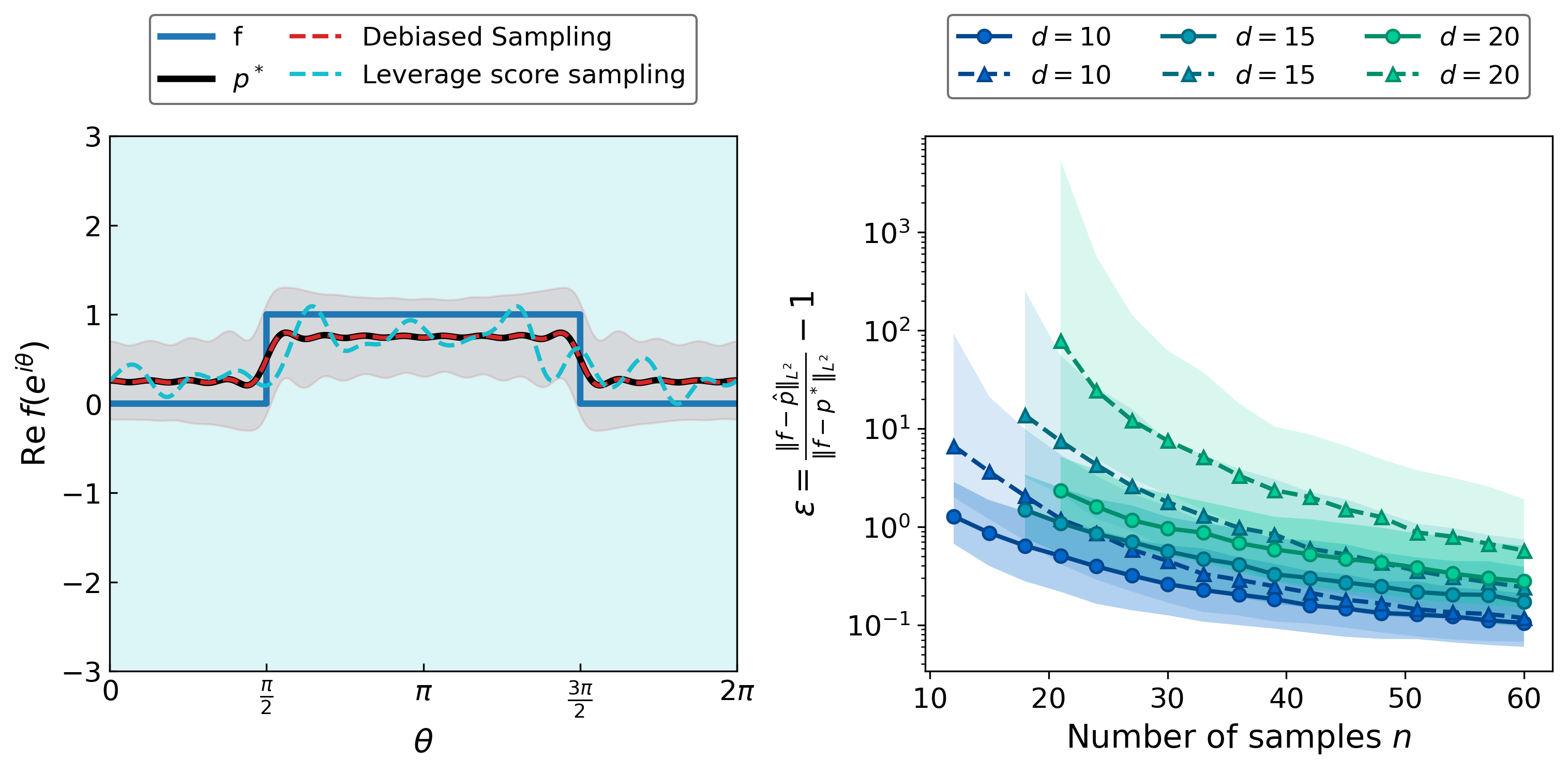}
    \caption{\textbf{Debiased Fourier regression.}
    We run \cref{alg:debiased-fourier} and leverage score sampling (i.e. uniform sampling) to find a Fourier series approximation of the function \(f(e^{i\theta}) = \mathbbm1_{\theta\in[3\pi/4, 5\pi/4]}\).\\
    \textit{Left:}
        We fix \(d=15\) with \(n=25\), and plot the resulting Fourier series \(\hat p\) for \cref{alg:debiased-fourier} (red) and leverage score sampling (blue), alongside the best-fit Fourier series \(p^*\) (black).
        Dotted lines show the average of returned series across 100,000 trials.
        Shaded areas show \(\pm 1\) standard deviation.
        Our method exactly recovers the best-fit series on average, while leverage score sampling does not appear to.
        Further, we have much lower variance than leverage score sampling.
        It is possible that, with sufficiently many trials, the dotted blue line may overlap with the black one.\\
    \textit{Right:}
        For number of Fourier coefficients \(d\in\{10,20,30\}\), we show how the relative error \(\eps_{empirical}\) (see \cref{sec:experiments}) of \cref{alg:debiased-generic} (circles \(\circ\)) and leverage score sampling (triangles \(\triangle\)) varies with the sample complexity \(n\).
        Solid lines show median error, with the shaded region showing the \(10\%\) and \(90\%\) quantiles across 1,000 trials.
        For small \(n\), say \(n \leq 3d\), we have an especially large improvement in relative error.
        As \(n\) grows larger, the difference between the methods dissipates.
}
    \label{fig:fourier}
\end{figure}
\begin{algorithm}[t]
  \caption{Debiased active Fourier regression} 
  \begin{algorithmic}[1] 
    \Require Oracle access to $f$, target degree-$d$, sample complexity \(n \geq d+1\)
    \Ensure Degree-$d$ polynomial $\hat p$
    \State Sample uniformly (i.e. Haar) random unitary matrix $\mX \in\bbC^{d+1 \times d+1}$ 
    \State Compute eigenvalues \(z_1,\ldots,z_{d+1}\in\cC\) of \mX
    \State Sample $z_{d+2}, \dots, z_{n} \in \cC$ iid uniformly from $\bbC$
    \State Evaluate \(f(z_i)\) for all \(i\in[n]\)
    \State Return polynomial \(\hat p = \argmin_{\deg(p) \leq d} \sum_{i=1}^{n}|p(z_i)-f(z_i)|^2\)
  \end{algorithmic}
  \label{alg:debiased-fourier}
\end{algorithm}

\begin{proof}[Proof of \cref{thm:debiased-fourier}]

As in the proof of \cref{thm:debiased-generic}, we reformulate \cref{prob:fourier} as a linear regression problem:
\[
    \min_{\deg(p) \leq d} \norm{p - f}_{\cC}^2
    =
    \min_{\vx\in\bbC^{d+1}}
    \E_{z \sim \mu}
    \bigl[|\va_{z}^*\vx - f(z)|^2\bigr]
\]
where \(\mu\) is the uniform distribution over \cC and $\va_{z} = [1 ~ z ~ z^2 ~ \cdots ~ z^d]^\top\in\bbC^{d+1}$.
The definitions of the leverage function, distribution, and projection DDP (\cref{def:levs,def:dpp}), as well as  \cref{impthm:leveraged-dpp} all extend to this complex setting without any alteration.
To complete the proof, we need to show that the eigenvalues computed on line 2 of \cref{alg:debiased-fourier} are indeed the projection DPP for \(\mu\) and that the uniformly random points sampled on line 3 of \cref{alg:debiased-fourier} are indeed sampled from the leverage score distribution for \(\mu\).

First, recalling that \(\mu\) is the uniform distribution on \cC, we find that the projection DPP for a set of points \(\cS = \{z_1, \ldots, z_{d+1}\}\) has density
\(q(\cS) \propto |\det(\mA_\cS)|^2\) where \(\mA_{\cS} = [\va_{z_1} ~ \cdots ~ \va_{z_{d+1}}]\in\bbC^{d+1 \times d+1}\).
Expanding this density via the generalized Vandermonde determinant formula \cite[Eqn.~7.3]{livan2018introduction}, we get
\(
    q(\cS) \propto \prod_{i < j} |z_i - z_j|^2.
\)
Since this pdf exactly matches the density in \cref{impthm:weyl-integral-fourier}, we conclude that the eigenvalues of a uniformly random unitary matrix form a projection DPP, and hence that the eigenvalues computed on line 2 of \cref{alg:debiased-fourier} are correct.

Further, Zandieh \etal \cite[Lem.~5]{zandieh2023near} show that the leverage function for the uniform distribution on \cC is constant on \cC.
Therefore, we find that the leverage score distribution for \(\mu\) on \cC is the uniform distribution on \cC, and hence that the leverage score samples generated on line 3 of \cref{alg:debiased-fourier} are correct.
\end{proof}

\section{Experimental Results}
\label{sec:experiments}

In this section we examine the empirical performance of \cref{alg:debiased-generic} under the uniform and Gaussian distributions and highlight several practical implementation choices that lead to a fast and numerically stable solver.
We first validate our claim that \cref{thm:debiased-generic} is unbiased.
We then demonstrate that in the low sample complexity regime (e.g.\ $n \leq 3d$), debiased sampling achieves significantly lower approximation error than leverage score sampling.
All experiments were run in Python v3.12 on the UC Berkeley SCF computing cluster, using 64 cores of a dual-socket AMD EPYC 7543 (2.8 GHz).

\subsection{Numerical demonstrations of \cref{thm:debiased-generic}}
In \cref{fig:unbiased} we empirically realize \cref{thm:debiased-generic} by studying polynomial approximation of the indicator function $$f(t)=\mathbbm{1}_{t\in[a,b]} \quad \text{with } a,b\in \mathbb{R},$$
where we pick \(a,b\) based on the distribution \(\mu\).
We choose the indicator function because it is known to be a hard instance for \cref{prob:bayesian} \cite{meyer2023near}.

In \cref{fig:unbiased}, we aggregate 100,000 independent trials to plot the mean degree-$15$ polynomial $\hat p$ produced by \cref{alg:debiased-generic} and the corresponding leverage score estimate, both shown alongside the optimal polynomial $p^*$ computed using standard quadrature methods.
In both examples, $\hat p$ produced by \cref{alg:debiased-generic} aligns with $p^*$ whereas the approximation produced by leverage score sampling struggles to match \(p^*\), or in the case of the uniform measure, fails entirely with the provided number of samples. 

\begin{figure}[t]
    \centering
    \includegraphics[width=1\linewidth]{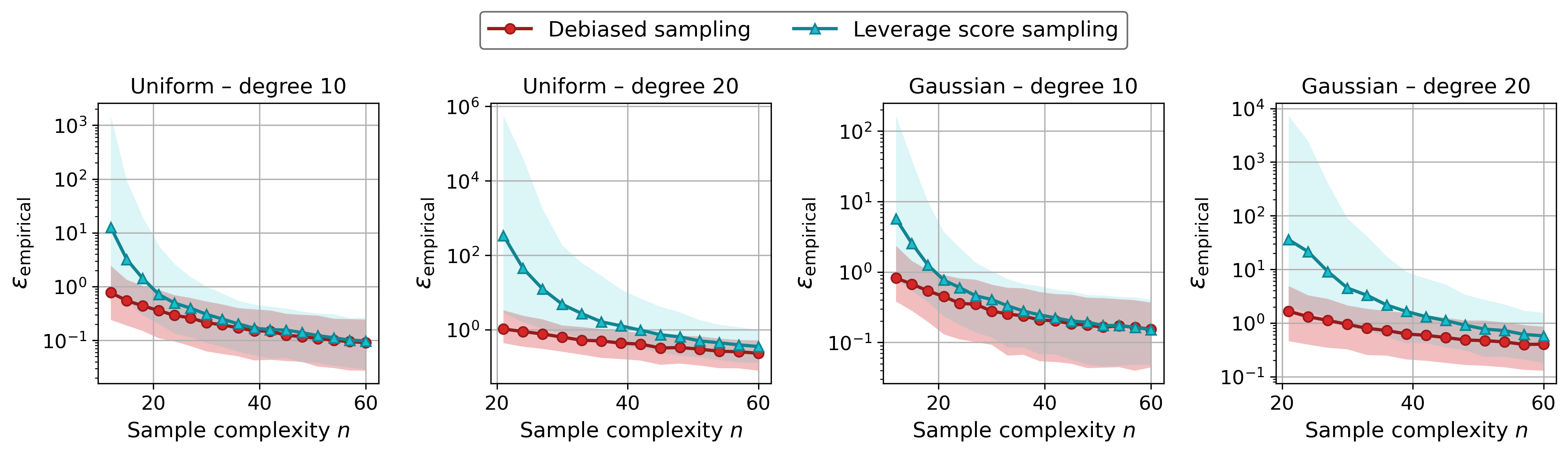}
    \caption{
    We run our debiased method \cref{alg:debiased-generic} (red) and leverage score sampling (blue) to find polynomial approximations of indicator functions under different distributions \(\mu\), across a variety of sample complexities \(n\), repeated for 1,000 trials.
    The first and third plots use degree \(d=10\), while the second and fourth plots use degree \(d=30\).    
We report the relative error \(\eps_{empirical}\) (see \cref{sec:experiments}).
    \textit{First two plots:}
    When \(\mu\) is the uniform distribution on \([-1,1]\), we take \(f(t) = \mathbbm1_{t\in[-0.5,0.5]}\).
    \textit{Last two plots:}
    When \(\mu\) is the Gaussian distribution \(\cN(0,1)\), we take \(f(t) = \mathbbm1_{t\in[-1,1]}\).
    Our method outperforms leverage score sampling in all regimes, and has an especially large edge when we have low sample complexity (say, \(n \leq 3d\)).
    Shaded areas show \(10\%\) and \(90\%\) quantiles.
    }
    \label{fig:errors}
\end{figure}

This fact is further demonstrated in \cref{fig:errors}, where we calculate the \textit{empirical relative error }
$$\varepsilon_{\rm empirical}=\frac{\E_{t\sim\mu}\bigl[|\hat p(t) - f(t)|^2\bigr]}{\E_{t\sim\mu}\bigl[|p^*(t) - f(t)|^2\bigr]}-1,$$
acquired by rearranging \cref{prob:bayesian}.
We investigate polynomial approximations of the same indicator functions using degree $d=10,15,20$ approximations, and study the behavior of this metric as we increase the number of samples for each method.
In the low sample complexity setting  (say, \(n \leq 3d\)), debiased sampling produces median relative errors that are $10-17\times$ more accurate than leverage score sampling using the same number of samples.
Notably, as \(n\) grows, all but a vanishing fraction of the samples used by \cref{alg:debiased-generic} become iid samples from the leverage score distribution.
So, in the large \(n\) regime, we see the performance of leverage score sampling and \cref{alg:debiased-generic} coincide.

Next, recall \cref{prob:fourier}, which extends the debiased sampling problem to approximating periodic functions via truncated Fourier series.
In \cref{sec:fourier} we show that this problem leads to \cref{alg:debiased-fourier} which we numerically investigate in \cref{fig:fourier}.
In this experiment we consider the complex analog of the previous experiment, constructing a periodic indicator function of the form
$$
f(e^{i\theta}) = \mathbbm1_{\theta\in[\pi/2,3\pi/2]}\quad \text{for } \theta \in [0,2\pi]$$
defined on the complex unit circle $\cC$.
In this setting as well, debiased sampling produces an accurate approximation $\hat p$, requiring less samples than iid leverage score sampling. 

\subsection{Implementation details}
To implement our methods effectively, we emphasize two design choices that lead to a fast and numerically stable version of \cref{alg:debiased-generic}.

First, to ensure that our eigenvalue sampler \cref{alg:eigval-sampler} runs in \(\cO(d \log d)\) time, it is essential to use a \emph{symmetric tridiagonal eigenvalue solver}, such as the \texttt{scipy.linalg.eigvalsh\_tridiagonal} subroutine in \texttt{scipy}.
Second, we must solve the polynomial approximation \(\hat p\) from line 5 of \cref{alg:debiased-generic} in a numerically stable way.
To do so, we expand \(\hat p(t) = \sum_{i=0}^d x_{i+1}\rP_i(t)\) as a linear combination of polynomials \(\rP_0,\ldots,\rP_d\) and solve for the coefficient vector \(\vx\in\bbR^{d+1}\):
\[
    \min_{\deg(p)\leq d} \sum_{i=1}^n \frac{1}{\tau(t_i)}\bigl| p(t_i) - f(t_i) \bigr|^2
    = \min_{\vx\in\bbR^{d+1}} \norm{\mS(\mV\vx-\vb)}_2^2
\]
where \(v_{ij}=\rP_j(t_i)\) is a generalized Vandermonde matrix, \(s_{ii} = 1/\sqrt{\tau(t_i)}\) is diagonal, and \(b_i = f(t_i)\).
This is most commonly done with the monomials \(\rP_j(t)=t^j\), though this is known to be unstable for large $n$ \cite{gautschi1987lower,gautschi2020stable}.
Our implementation instead takes \(\rP_0,\ldots,\rP_d\) to be the orthonormal polynomials for \(\mu\).

\pad
\noindent
Full implementations of the described algorithms in this work are available publicly at
\actionbox{\href{https://github.com/chriscamano/Debiased-Polynomial-Regression}{https://github.com/chriscamano/Debiased-Polynomial-Regression}}

\section{Acknowledgements}
Chris Camaño was supported by the Kortschak scholars program, and the National Science Foundation Graduate Research Fellowship under Grant No. 2139433.
Raphael Meyer was partially supported by a Caltech Center for Sensing to Intelligence grant to Joel A. Tropp and ONR Award N-00014-24-1-2223 to Joel A. Tropp.
We thank the Berkeley Statistical Computing Facility for access to computational resources.
We thank Ethan Epperly for detailed comments improving the presentation of the paper.
We also thank Jorge Garza-Vargas for helpful discussion.

\bibliographystyle{siamplain}

{\footnotesize
\begin{thebibliography}{10}

\bibitem{adcock2024optimal}
{\sc B.~Adcock}, {\em Optimal sampling for least-squares approximation}, arXiv
  preprint arXiv:2409.02342,  (2024),
  \url{https://doi.org/10.48550/arXiv.2409.02342}.

\bibitem{adcock2018compressed}
{\sc B.~Adcock, S.~Brugiapaglia, and C.~G. Webster}, {\em Compressed sensing
  approaches for polynomial approximation of high-dimensional functions}, in
  Compressed Sensing and its Applications: Second International MATHEON
  Conference 2015, Springer, 2018, pp.~93--124,
  \url{https://doi.org/10.1007/978-3-319-69802-1_3}.

\bibitem{adcock2024learning}
{\sc B.~Adcock, M.~Griebel, and G.~Maier}, {\em Learning lipschitz operators
  with respect to gaussian measures with near-optimal sample complexity}, arXiv
  preprint arXiv:2410.23440,  (2024), \url{https://arxiv.org/pdf/2410.23440}.

\bibitem{adcock2023fast}
{\sc B.~Adcock and A.~Shadrin}, {\em Fast and stable approximation of analytic
  functions from equispaced samples via polynomial frames}, Constructive
  Approximation, 57 (2023), pp.~257--294,
  \url{https://doi.org/10.1007/s00365-022-09593-2}.

\bibitem{avron2019universal}
{\sc H.~Avron, M.~Kapralov, C.~Musco, C.~Musco, A.~Velingker, and A.~Zandieh},
  {\em A universal sampling method for reconstructing signals with simple
  fourier transforms}, in Proceedings of the 51st Annual ACM SIGACT Symposium
  on Theory of Computing, 2019, pp.~1051--1063,
  \url{https://doi.org/10.1145/3313276.3316363}.

\bibitem{brubeck2021vandermonde}
{\sc P.~D. Brubeck, Y.~Nakatsukasa, and L.~N. Trefethen}, {\em Vandermonde with
  arnoldi}, Siam Review, 63 (2021), pp.~405--415,
  \url{https://doi.org/10.1137/19M130100X}.

\bibitem{byun2025progress}
{\sc S.-S. Byun and P.~J. Forrester}, {\em Progress on the Study of the Ginibre
  Ensembles}, Springer Nature, 2025,
  \url{https://doi.org/10.1007/978-981-97-5173-0}.

\bibitem{chen2019active}
{\sc X.~Chen and E.~Price}, {\em Active regression via linear-sample
  sparsification}, in Conference on Learning Theory, PMLR, 2019, pp.~663--695,
  \url{https://proceedings.mlr.press/v99/chen19a.html}.

\bibitem{coakley13}
{\sc E.~S. Coakley and V.~Rokhlin}, {\em A fast divide-and-conquer algorithm
  for computing the spectra of real symmetric tridiagonal matrices}, Applied
  and Computational Harmonic Analysis, 34 (2013), pp.~379--414,
  \url{https://doi.org/10.1016/j.acha.2012.06.003}.

\bibitem{derezinski2021determinantal}
{\sc M.~Derezinski and M.~W. Mahoney}, {\em Determinantal point processes in
  randomized numerical linear algebra}, Notices of the American Mathematical
  Society, 68 (2021), pp.~34--45,
  \url{https://doi.org/10.48550/arXiv.2005.03185}.

\bibitem{derezinski2022unbiased}
{\sc M.~Derezi{\'n}ski, M.~K. Warmuth, and D.~Hsu}, {\em Unbiased estimators
  for random design regression}, Journal of Machine Learning Research, 23
  (2022), pp.~1--46, \url{https://jmlr.org/papers/v23/19-571.html}.

\bibitem{devroye23}
{\sc L.~Devroye and J.~Hamdan}, {\em A note on the exact simulation of a random
  eigenvalue of a gue matrix}, arXiv e-prints,  (2023), pp.~arXiv--2304,
  \url{https://doi.org/10.48550/arXiv.2304.03741}.

\bibitem{dumitriu2002matrix}
{\sc I.~Dumitriu and A.~Edelman}, {\em Matrix models for beta ensembles},
  Journal of Mathematical Physics, 43 (2002), pp.~5830--5847,
  \url{https://doi.org/10.1063/1.1507823}.

\bibitem{duy2018spectral}
{\sc T.~K. Duy}, {\em On spectral measures of random jacobi matrices}, Osaka
  Journal of Mathematics, 55 (2018), pp.~595--617,
  \url{https://doi.org/10.18910/70813}.

\bibitem{dyson1970correlations}
{\sc F.~J. Dyson}, {\em Correlations between eigenvalues of a random matrix},
  Communications in Mathematical Physics, 19 (1970), pp.~235--250,
  \url{https://doi.org/10.1007/BF01646824}.

\bibitem{erdelyi2020}
{\sc T.~Erd{\'e}lyi, C.~Musco, and C.~Musco}, {\em Fourier sparse leverage
  scores and approximate kernel learning}, Advances in Neural Information
  Processing Systems, 33 (2020), pp.~109--122,
  \url{https://papers.nips.cc/paper/2020/hash/012d9fe15b2493f21902cd55603382ec-Abstract.html}.

\bibitem{fasi2021sampling}
{\sc M.~Fasi and L.~Robol}, {\em Sampling the eigenvalues of random orthogonal
  and unitary matrices}, Linear Algebra and its Applications, 620 (2021),
  pp.~297--321, \url{https://doi.org/10.1016/j.laa.2021.02.031}.

\bibitem{gautschi2020stable}
{\sc W.~Gautschi}, {\em How (un) stable are vandermonde systems?}, in
  Asymptotic and computational analysis, CRC Press, 2020, pp.~193--210,
  \url{https://www.cs.purdue.edu/homes/wxg/selected_works/section_01/118.pdf}.

\bibitem{gautschi1987lower}
{\sc W.~Gautschi and G.~Inglese}, {\em Lower bounds for the condition number of
  vandermonde matrices}, Numerische Mathematik, 52 (1987), pp.~241--250,
  \url{https://doi.org/10.1007/BF01398878}.

\bibitem{guo2020constructing}
{\sc L.~Guo, A.~Narayan, and T.~Zhou}, {\em Constructing least-squares
  polynomial approximations}, SIAM Review, 62 (2020), pp.~483--508,
  \url{https://doi.org/10.1137/18M1234151}.

\bibitem{hampton2015coherence}
{\sc J.~Hampton and A.~Doostan}, {\em Coherence motivated sampling and
  convergence analysis of least squares polynomial chaos regression}, Computer
  Methods in Applied Mechanics and Engineering, 290 (2015), pp.~73--97,
  \url{https://doi.org/10.1016/j.cma.2015.02.006}.

\bibitem{kane2017robust}
{\sc D.~Kane, S.~Karmalkar, and E.~Price}, {\em Robust polynomial regression up
  to the information theoretic limit}, in 2017 IEEE 58th Annual Symposium on
  Foundations of Computer Science, IEEE, 2017, pp.~391--402,
  \url{https://doi.org/https://doi.org/10.1109/FOCS.2017.43}.

\bibitem{killip2004matrix}
{\sc R.~Killip and I.~Nenciu}, {\em Matrix models for circular ensembles},
  International Mathematics Research Notices, 2004 (2004), pp.~2665--2701,
  \url{https://doi.org/10.1155/S1073792804141597}.

\bibitem{livan2018introduction}
{\sc G.~Livan, M.~Novaes, and P.~Vivo}, {\em Introduction to random matrices
  theory and practice}, Monograph Award, 63 (2018), p.~914,
  \url{https://doi.org/https://doi.org/10.1007/978-3-319-70885-0}.

\bibitem{meckes2019random}
{\sc E.~S. Meckes}, {\em The random matrix theory of the classical compact
  groups}, vol.~218, Cambridge University Press, 2019,
  \url{https://doi.org/10.1017/9781108303453}.

\bibitem{meyer2023near}
{\sc R.~A. Meyer, C.~Musco, C.~Musco, D.~P. Woodruff, and S.~Zhou}, {\em
  Near-linear sample complexity for {Lp} polynomial regression}, in Proceedings
  of the 2023 Annual ACM-SIAM Symposium on Discrete Algorithms, SIAM, 2023,
  pp.~3959--4025, \url{https://doi.org/10.1137/1.9781611977554.ch153}.

\bibitem{musco2022active}
{\sc C.~Musco, C.~Musco, D.~P. Woodruff, and T.~Yasuda}, {\em Active linear
  regression for $\ell_p$ norms and beyond}, in 2022 IEEE 63rd Annual Symposium
  on Foundations of Computer Science, IEEE, 2022, pp.~744--753,
  \url{https://doi.org/10.1109/FOCS54457.2022.00076}.

\bibitem{pauwels2018relating}
{\sc E.~Pauwels, F.~Bach, and J.-P. Vert}, {\em Relating leverage scores and
  density using regularized christoffel functions}, in Advances in Neural
  Information Processing Systems, S.~Bengio, H.~Wallach, H.~Larochelle,
  K.~Grauman, N.~Cesa-Bianchi, and R.~Garnett, eds., vol.~31, Curran
  Associates, Inc., 2018,
  \url{https://proceedings.neurips.cc/paper_files/paper/2018/file/aff1621254f7c1be92f64550478c56e6-Paper.pdf}.

\bibitem{schwab2021deep}
{\sc C.~Schwab and J.~Zech}, {\em Deep learning in high dimension: Neural
  network approximation of analytic functions in
  ${L}^2(\mathbb{R}^d,\gamma_d)$}, arXiv preprint arXiv:2111.07080,  (2021),
  \url{https://arxiv.org/pdf/2111.07080}.

\bibitem{shimizu24}
{\sc A.~Shimizu, X.~Cheng, C.~Musco, and J.~Weare}, {\em Improved active
  learning via dependent leverage score sampling}, in The Twelfth International
  Conference on Learning Representations, 2024,
  \url{https://openreview.net/forum?id=IYxDy2jDFL}.

\bibitem{shustin2022semi}
{\sc P.~F. Shustin and H.~Avron}, {\em Semi-infinite linear regression and its
  applications}, SIAM Journal on Matrix Analysis and Applications, 43 (2022),
  pp.~479--511, \url{https://doi.org/10.1137/21M1411950}.

\bibitem{tang2014discrete}
{\sc T.~Tang and T.~Zhou}, {\em On discrete least-squares projection in
  unbounded domain with random evaluations and its application to parametric
  uncertainty quantification}, SIAM Journal on Scientific Computing, 36 (2014),
  pp.~A2272--A2295, \url{https://doi.org/https://doi.org/10.1137/140961894}.

\bibitem{trefethen2019approximation}
{\sc L.~N. Trefethen}, {\em Approximation theory and approximation practice,
  extended edition}, SIAM, 2019, \url{https://doi.org/10.1137/1.9781611975949}.

\bibitem{zandieh2023near}
{\sc A.~Zandieh, I.~Han, and H.~Avron}, {\em Near optimal reconstruction of
  spherical harmonic expansions}, Advances in Neural Information Processing
  Systems, 36 (2023), pp.~23968--23989,
  \url{https://proceedings.neurips.cc/paper_files/paper/2023/hash/4b719e74623f4fa238ded71b56f0a184-Abstract-Conference.html}.

\end{thebibliography}
 }

\end{document}